\documentclass[conference]{IEEEtran}
\usepackage[english]{babel}
\usepackage[utf8]{inputenc}
\usepackage{amsmath}
\usepackage{amsthm, amssymb, cite, graphicx, epsfig, url, color, setspace, epstopdf, subfigure}
\usepackage[colorinlistoftodos]{todonotes}
\usepackage[noend]{algorithmic}
\usepackage[ruled]{algorithm}
\usepackage{epstopdf}
\usepackage{amsthm,amssymb}
\usepackage{cases}
\usepackage{mathtools,breqn}
\usepackage[font=small]{caption}
\usepackage{epstopdf}

\newtheorem{prop}{Proposition}

\newcommand{\ignore}[1]{}

\title{Optimal Timing of Moving Target Defense: A Stackelberg Game Model}

\author{\IEEEauthorblockN{
Henger Li and
Zizhan Zheng\\
\IEEEauthorblockA{Department of Computer Science, Tulane University, New Orleans, USA}
Email: \{hli30, zzheng3\}@tulane.edu}}

\begin{document}

\maketitle

\begin{abstract}
As an effective approach to thwarting advanced attacks, moving target defense (MTD) has been applied to various domains. Previous works on MTD, however, mainly focus on deciding the sequence of system configurations to be used and have largely ignored the equally important timing problem. Given that both the migration cost and attack time vary over system configurations, it is crucial to jointly optimize the spatial and temporal decisions in MTD to better protect the system from persistent threats. In this work, we propose a Stackelberg game model for MTD where the defender commits to a joint migration and timing strategy to cope with configuration-dependent migration cost and attack time distribution. 
The defender's problem is  formulated as a semi-Markovian decision process and a nearly optimal MTD strategy is derived by exploiting the unique structure of the game.     
\end{abstract}
\section{Introduction}~\label{sec:intro}
Cyber-attacks are becoming increasingly more adaptive and sophisticated. One example is Advanced Persistent Threats (APTs)~\cite{APT}, an emerging class of continuous and stealthy hacking processes launched by incentive
driven entities. To avoid immediate detection and obtain long-term benefit, an advanced attacker may carefully cover its tracks, e.g., by internationally operating in a ``low-and-slow'' fashion~\cite{graph-apt}. The stealth and persistent nature makes these attacks extremely difficult to defense using traditional techniques that focus on one-shot attacks of known types. 

An important obstacle in combating stealthy attacks is {\it information asymmetry}. An advanced attack often involves an information collection stage (e.g., through probing the system) to dynamically identify the best target to attack. In contrast, a defender typically knows much less about a stealthy and adaptive attacker. To revert the information asymmetry, an promising approach is moving target defense (MTD), where the defender constantly updates the system configuration to increase the attacker's uncertainty. By exploiting the diversity and randomness at different system layers, various MTD techniques have been proposed including dynamic networks~\cite{jafarian2012openflow,IP-mapping-MTD14}, dynamic platforms~\cite{salamat2011runtime}, dynamic runtime environments~\cite{ASLR}, dynamic software~\cite{GenProg}, and dynamic data~\cite{nguyen2008security}. 

In addition to empirical evaluation of domain specific MTD techniques, decision and game theoretic approaches have recently been adopted to derive more cost-effective MTD solutions. In particular, a zero-sum dynamic game for MTD is proposed in~\cite{zhu2013game} where a fixed migration cost (i.e., the cost of switching from one configuration to another) is assumed. More recently, Bayesian Stackelberg games (BSG) have been applied to MTD~\cite{sengupta2017game}, where the
defender commits to an {\it i.i.d.} strategy independent of the real-time system configuration, leading to a suboptimal strategy. In our previous work~\cite{gamesec-17}, we have proposed a Markovian modeling of MTD where the decision on the next system configuration to be used depends on the current one and derived the optimal Stackelberg strategy.

An important limitation of existing game models for MTD, however, is that the {\it temporal} decision has been largely ignored. Previous studies mainly focus on the {\it spatial} decision, i.e., what is the next configuration to be used, while assuming a simplified decision on timing. In particular,  constant attack times and periodic migration policies are commonly assumed in previous work. 
However, different system configurations typically require different techniques and expertise to set up and to identify and exploit vulnerability. Thus, both the migration cost and the amount of time that the attacker needs to take down a system are configuration dependent. 
Both of them should be taken into consideration when deciding {\it when} to move as large attack times (or migration costs) imply less frequent updates. Therefore, a simple periodic migration strategy is far from satisfactory. 

In this paper, we make the first effort on the joint optimization of spatial and temporal decisions in MTD. We extend our Markovian modeling of MTD in~\cite{gamesec-17} by introducing 
attack times that are both random and configuration dependent. We consider a Stackelberg game model with the defender as the leader and the attacker as the follower. The defender commits to a stationary MTD strategy at the beginning of the game, which includes both a configuration transition matrix (spatial decision) and a set of defense periods, one for each configuration (temporal decision). Further, instead of minimizing the long-term discounted cost as in~\cite{gamesec-17}, we consider the more challenging time-average cost objective in this work, which is more reasonable for patient attackers targeting long-term advantages. The problem of finding the best strategy for the defender is formulated as a semi-Markovian decision process~\cite{Puterman-MDP} with continuous decision variables.
Although semi-MDPs with continuous decisions are difficult to solve in general, we show that the classic value iteration (VI) algorithm can be applied to our problem to obtain a nearly optimal stationary strategy. We further derive an efficient solution to the Min-Max problem in each iteration of VI by utilizing the unique structure of the MTD game. 

We have made the following contributions in this paper.
\begin{itemize}
    \item We propose a new active defense paradigm that incorporates {\it spatial} and {\it temporal} decisions to achieve robust moving target defense.
    \item We extend the Bayesian Stackelberg game (BSG) model by considering Markovian defense strategies, which are more general than the repeated decisions in BSG and are more appropriate for MTD. 
    \item We derive a nearly optimal defense strategy based on the value iteration technique and propose an efficient algorithm for each iteration by utilizing the unique structure of the Min-Max problem in the MTD game.
\end{itemize}

The rest of the paper is organized as follows. We review the related work on MTD games in Section~\ref{sec:related} and present the game model and problem formulation in
Section~\ref{sec:model}. The optimal defense strategy and its analysis are discussed in Section~\ref{sec:solution}. We evaluate our solution in Section~\ref{sec:eval} and conclude the paper in Section~\ref{sec:conclusion}.
\section{Related Work}\label{sec:related}
Several game theoretic models have been proposed for MTD in the last few years~\cite{mtdgame-survey-2017}. A zero-sum dynamic game for MTD is proposed in~\cite{zhu2013game}. 
Each player chooses its action independently in each round according to a mixed strategy, and gets immediate feedback on its payoff. 
However, the zero-sum assumption does not hold in many security scenarios. Further, a fixed migration cost is assumed, which neglects the heterogeneity in configurations.  
More recently, Bayesian Stackelberg games (BSG) have been applied to MTD in web applications~\cite{sengupta2017game}, where the defender commits to an $i.i.d.$ migration strategy, 
which is suboptimal when the migration cost is configuration dependent. A BSG model for the closely related cyber deception problem is studied in~\cite{haifeng-CDG}, which, however, considers a purely static setting. Several Markov models for MTD have also been proposed recently~\cite{zhuang2014model,maleki2016markov}. 
However, these works focus on analyzing the expected time needed for the attacker to compromise the system under simple defense strategies and do not consider how to derive optimal MTD strategies. 
In our recent work~\cite{gamesec-17}, we have extended the BSG models by introducing Markovian strategies into MTD while still considering periodic migrations as in previous works.
Initiated by the FlipIt game~\cite{flipit}, optimal timing of security updates has received a lot of interest recently~\cite{asymmetric-model,ming-gamesec2015,optimal-timing}, where instead of switching between configurations, the system is recovered after a certain time period. However, these studies do not apply to the optimal timing of MTD directly. 

\section{MTD Game Model and Problem Formulation}\label{sec:model}
In this section, we present our game theoretic model for MTD and formulate the defender's optimization problem.  Figure~\ref{fig:model} gives an example of our attack-defense model. 

\subsection{System Model}
\noindent{\bf System Configurations:} We consider a system to be protected and two players, an attacker and a defender. The system has a set of configuration parameters that the defender can choose from. Examples include IP addresses, network topology, OS versions, memory address space layout, etc. To meet the system’s integrity and performance requirement, only a subset of configurations is valid, which is defined as the system configuration space, denoted by $S$. Let $n = |S|$ denote the number of configurations. 
\begin{table}
\begin{center}
\small{
\begin{tabular}{|p{0.07\linewidth}|p{0.78\linewidth}|}
 \hline
 $S$ & set of all configurations\\
 \hline
 $n$ & number of configurations\\
 \hline
 $a_j$ & random attack time for  configuration $j$\\
 \hline
 $p_{ij}$ & transition probability from configuration $i$ to $j$\\
 \hline
 $P$ & transition probability matrix\\
 \hline
 $\mathbf{p}_i$ & the $i$-th row in $P$\\
 \hline
 $\alpha$ & lower bound of $p_{ij}$\\
 \hline
 $m_{ij}$ & migration cost from configuration $i$ to $j$\\
 \hline
 $M$ & migration cost matrix\\
 \hline
 $\tau_{i}$ & the next defense period when the current configuration is $i$\\
 \hline
 $\overline{\tau}$, $\underline{\tau}$ & maximum/minimum defense period\\
 \hline
 $\gamma$ & a parameter in transforming semi-MDP to MDP\\
 \hline
 $\delta$ & step size in searching for $\tau$\\
 \hline
 $\omega$ & a parameter controlling the stopping criterion in Algorithm 1\\
 \hline
\end{tabular}
}
\caption{List of symbols in the paper}
\end{center}
\end{table}

\vspace{1ex}
\noindent{\bf Defense Model:} The defender constantly migrates the system configuration to increase the attacker's uncertainty. We assume that a migration happens instantaneously subject to a cost $m_{ij}$ if the system moves from configuration $i$ to configuration $j$. We allow $m_{ii}>0$ to model the cost of recovering the system to the same configuration. Let $M$ denote the matrix of migration costs $\{m_{ij}\}_{n \times n}$. A continuous time horizon is considered. Let $t_k$ denote the time instance when the $k$-th migration happens  
and $s_k$ the system configuration in the $k$-th defense period (from $t_{k-1}$ to $t_k$). At the end of the $k$-th defense period, the defender picks the next configuration $s_{k+1}$ with probability $p_{s_ks_{k+1}}$. We assume $t_0 = 0$ and let $s_0$ denote the initial configuration (before $t_0$).

We assume that the defender adopts a {\it stationary} strategy consisting of (1) a transition matrix $P = \{p_{ij}\}_{n \times n}$ where $p_{ij}$ is the probability of moving to configuration $j$ when the system is currently in configuration $i$, and (2) a vector $\{\tau_i\}_{i \in S}$, where $\tau_i$ is the next defense period to be used if the system is currently in configuration $i$. According to this definition, the $k-th$ defense period only depends on $s_{k-1}$ but not $s_k$ (see Figure~\ref{fig:model} for an example). This is to simplify the decision problem as we discuss below. We may also consider strategies where $\tau_k$ depends on $s_k$ only or both $s_{k-1}$ and $s_{k}$, which is left to our future work. Without loss of generality, we assume that $\tau_i \in [\underline{\tau},\overline{\tau}]$ for any $i$ where $\underline{\tau}>0$ and $\overline{\tau} < \infty$. Let $\mathbf{p}_i$ denote the $i$-th row of $P$.

\vspace{1ex}
\noindent{\bf Attack Model:} We consider a persistent attacker that continuously probes and attacks the system. We assume that once a migration happens, the attacker learns this fact immediately and makes a guess on the new configuration. Further, the amount of time needed to compromise the system under configuration $j$ is modeled as a random variable $a_j$ with distribution $A_j$ and is $i.i.d.$ across attacks. Consider the $k$-th defense period. Let $\hat{s}_k$ denote the attacker's guess of $s_k$. Under the stationary defense strategy described above, the probability that the attacker's guess is correct is $\text{Pr}(\hat{s}_k = s_k) = p_{s_{k-1}\hat{s}_k}$. The expected amount of time that the system is compromised in the $k$-th period then becomes 
$p_{s_{k-1}\hat{s}_k} \mathbb{E}[\max(\tau_{s_{k-1}} - a_{\hat{s}_k},0)]$ where the expectation is with respect to the randomness of attack time. 

\begin{figure}[t]
\centering
\includegraphics[width=\linewidth]{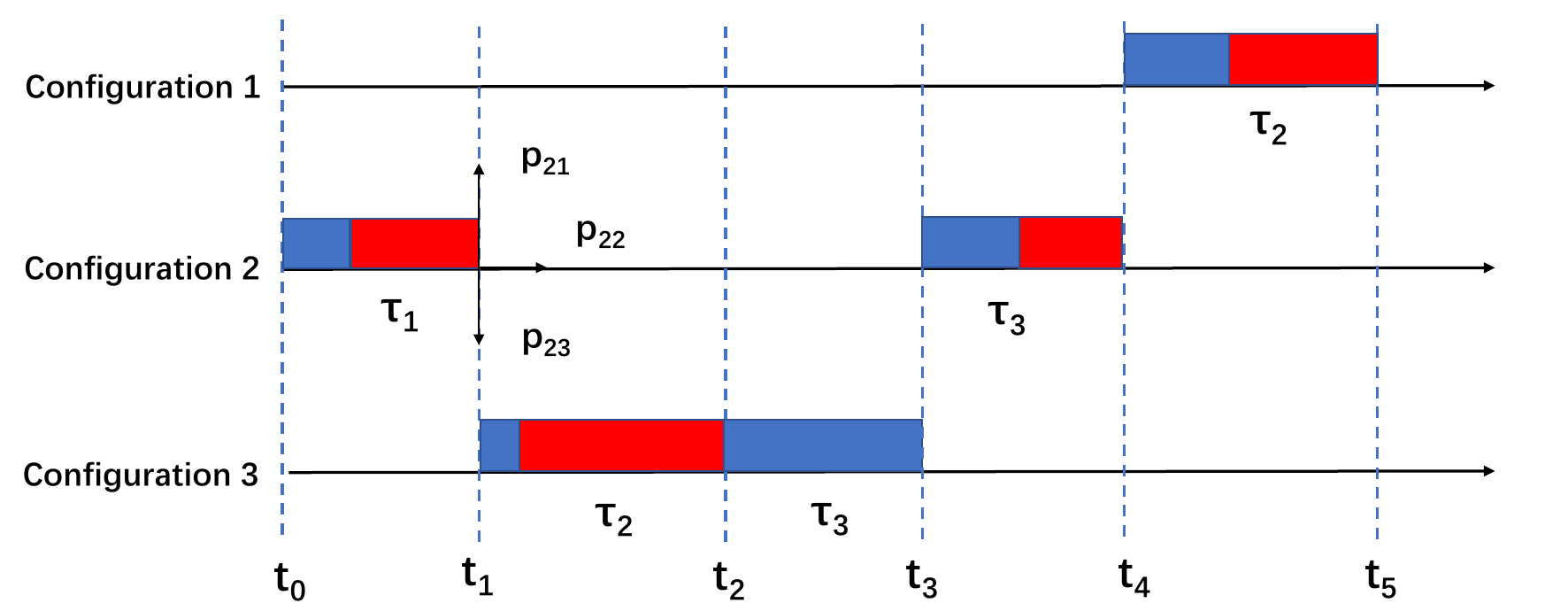}
\caption{\small An example of the game model where configuration 1 is the initial configuration. A blue (resp. red) block denotes a time interval when the system is protected (resp. compromised). $\tau_i$ is the length of the current defense period when the previous configuration is $i$.}
\label{fig:model}
\end{figure}

\vspace{1ex}
\noindent{\bf Stackelberg Game:} We assume that the attacker always learns $s_k$ at the end of the $k$-th defense period (a worst-case scenario from the defender’s perspective). Consequently, the attacker may also learn the defender’s stationary strategy once enough samples are collected. To simplify the analysis, we assume that the defender announces its strategy in the beginning of the game. We further assume that the attacker is myopic and always exploits the most beneficial configuration according to the defender’s
strategy and the previous system configuration it observed. We then have $\hat{s}_{k} = \text{argmax}_{j \in S} p_{s_{k-1}j} \mathbb{E}[\max(\tau_{s_{k-1}} - a_j,0)]$. Effectively, we are considering a Stackelberg game with the defender as the leader and the attacker as the follower. As is typical in security games, we assume that the defender knows the attack time distribution (but not its realization). 
It is important to note that our game model is more general than the Bayesian Stackelberg Game (BSG) models in~\cite{sengupta2017game, haifeng-CDG} since the leader (defender) commits to a configuration-dependent Markovian strategy rather than a simple ${\it i.i.d.}$ strategy as in BSG models where $p_{ij}$ is a constant across $i$. 
To simplify
the notation, we let $w_{ij} = \mathbb{E}[\max(\tau_i - a_j,0)]$ for $i,j \in S$.

\subsection{Defender's Problem as a Semi-MDP}
The defender's objective is to strike a balance between the loss from attacks and the cost of migration. To this end, we formulate the defender's problem as an average-cost semi-MDP as follows. We define the state of the system as the set of configurations $S$. Let $s_0$ be the initial state (before the game begins). We consider stationary policies only. Each time the system is in state $i$, a control $\mu(i) \triangleq (\mathbf{p}_i,\tau_i)$ is applied, and the defender incurs an expected cost  $c(i,\mu(i)) = \max_{j}(p_{ij}w_{ij}) + \sum_{j} p_{ij}m_{ij}$, and the system moves to state $j$ with probability $p_{ij}$. Note that the cost function includes both the expected loss from attacks as well as the expected migration cost. For a given policy $\mu$, the time-average cost of the defender starting from an initial state $s_0$ is defined as: 
\begin{align}
C_\mu(s_0) &= \limsup_{N\rightarrow\infty}\frac{\sum^{N-1}_{k=0} c(s_k,\mu(s_k))}{\sum^{N-1}_{k=0}{\tau_{s_k}}} \nonumber \\
&=\limsup_{N\rightarrow\infty}\frac{\sum^{N-1}_{k=0} [\max_{j}(p_{s_kj}w_{s_kj}) + \sum_{j} p_{s_kj}m_{s_kj}]}{\sum^{N-1}_{k=0}{\tau_{s_k}}} \label{semi-mdp}
\end{align}

The defender's goal is to commit to a policy $\mu$ that minimizes its time-average cost for any initial state. We thus obtain an infinite-horizon semi-MDP with average cost criterion and continuous decision variables.

\section{Optimal MTD Strategies}\label{sec:solution}
In this section, we propose efficient algorithms to find a nearly optimal solution to the defender's problem. 

\subsection{Approximation and Transformation}
The semi-MDP defined in~\eqref{semi-mdp} has a finite state space $S$ and a compact action space $[0,1]^{n} \times [\underline{\tau},\bar{\tau}]$. There are two major challenges to solve the semi-MDP. First, when an arbitrary transition matrix $P$ is allowed, the Markov chain associated with a given stationary policy is not necessarily unichain. Consequently, the time-average cost may vary over the initial configurations~\cite{Puterman-MDP}. Second, the continuous decision variables make it challenging to apply standard techniques such as value iteration and policy iteration as the Min-Max problem (defined below) in each iteration can be difficult to solve. 

We discuss how to address the second challenge in the next subsection. To address the first challenge, we impose the following constraint on $P$ by requiring that $p_{ij} \geq \alpha$ for any $i,j$ where $\alpha>0$ is a small number. 
With this simple constraint, the unichain requirement is always satisfied. To understand why the assumption is reasonable, consider two stationary policies with transition matrices $P$ and $P'$, respectively, both of which are unichain. If $|P_{ij}-P'_{ij}| \leq \alpha$ for any $i,j$, then the stationary distribution of $P$ is close to that of $P'$ by making $\alpha$ small enough~\cite{perturbation-mdp}. Thus, the loss of optimality is negligible for small enough $\alpha$ if we only consider unichain policies. On the other hand, in a multichain policy, some configurations are never used for MTD, which is unlikely to happen in practice as 
it reduces the attacker's uncertainty. 
A rigorous understanding of the multichain case is left to our future work. 

With the above assumption and the fact that the single stage cost $c(i,\mu(i))$ is continuous in $p_{ij}$ and $\tau_i$, it is known that the optimal time-average cost is independent of the initial configuration, and further, there is a stationary deterministic policy that is optimal~\cite{Puterman-MDP}. Moreover, we can apply a standard trick to transform the semi-MDP to a discrete-time MDP with average cost criterion defined below: 
\begin{align}
\tilde{C}_\mu(s_0) &= \limsup_{N\rightarrow\infty}\frac{1}{N}\sum^{N-1}_{k=0} \tilde{c}(s_k,\mu(s_k)) 
\label{mdp}
\end{align}

\noindent with the single stage cost and transition probabilities given by~\cite{Puterman-MDP}:
\begin{align}
\tilde{c}(i,\mu(i)) &= \frac{\max_{j}{(w_{ij}p_{ij})} + \sum_{j}{ p_{ij}m_{ij}}}{\tau_i}    \\
\tilde{p}_{ij} &= \gamma \frac{p_{ij}-\delta_{ij}}{\tau_i} + \delta_{ij}  
\end{align}

\noindent where $\delta_{ij} = 1$ if $i = j$ and $\delta_{ij}=0$ otherwise, and $\gamma$ satisfies $0 < \gamma < \tau_i/(1 - p_{ii})$ for any $i \in S$ and $p_{ii}<1$. We choose $\gamma = \underline{\tau}$ in this work. The original semi-MDP and the transformed discrete-time MDP have the same class of stationary policies. Further, for each stationary policy $\mu$, $C_{\mu}(i) = \tilde{C}_{\mu}(i)$ for any $i\in S$. This result does not require any assumption about the chain structures of the Markov chains associated with the stationary policies.

\subsection{Value Iteration Algorithm}
Given the transformation defined above, the problem then boils down to solving the average cost MDP in~\eqref{mdp}. Since $\tilde{p}_{ij}>0$ for any $i,j$, the MDP is still unichain. As the state space is finite and the action space is a separable metric space, it is known that the standard value iteration algorithm converges to the optimal average cost~\cite{vi-average-MDP}. The main challenge is to design an efficient solution to the Min-Max problem in each iteration as discussed below. 

The VI algorithm (see Algorithm 1) maintains a value vector $V^t \in \mathbb{R^+}^{n}$ in each iteration $t$. Initially, $V^0(i) = 0$ for each $i \in S$. In iteration $t$, the algorithm solves the following Min-Max problem for each configuration $i \in S$ (lines 4-11):
\begin{align}
V^t(i) &= \min_{\mathbf{p}_i,\tau_i}\Big[\tilde{c}(i,\mu(i))+\sum_{j\in S}\tilde{p}_{ij}V^{t-1}(j)\Big] \nonumber \\
&= \min_{\mathbf{p}_i,\tau_i}\Bigg[
\frac{\max_{j}{(w_{ij}p_{ij})} + \sum_{j}{p_{ij}(m_{ij}+\gamma V^{t-1}(j))}}{\tau_i} \nonumber \\
&\hspace{25ex}+(1-\frac{\gamma}{\tau_i})V^{t-1}(i) \Bigg] \nonumber \\
& \hspace{5ex} s.t. \hspace{2ex} \mathbf{p}_i \in [\alpha,1]^{n}, \sum_j p_{ij} = 1, \tau_i \in [\underline{\tau},\bar{\tau}].
\label{min-max}
\end{align}

Let $V^t(i,\mathbf{p},\tau)$ denote the value of the objective function in~\eqref{min-max} when $\mathbf{p}_i = \mathbf{p}$ and $\tau_i = \tau$. The Min-Max problem is difficult to solve due to the coupling of $P$ and $\tau$. To this end, we discretize the search space for $\tau_i$. For each $\tau_i \in \{\underline{\tau},\underline{\tau}+\delta,\underline{\tau}+2\delta,..., \bar{\tau}]$ where $\delta$ is a parameter, (5) is solved to search for the best $\mathbf{p}_i$ (lines 6-10). An efficient solution for this step is discussed below. A smaller $\delta$ gives a better solution at the expense of a higher searching overhead. 

Algorithm 1 stops when $V^t(i)-V^{t-1}(i)$ is close to a constant across $i$ (see lines 15-17 where $\omega$ is a parameter). When the algorithm stops, $V^t(i)-V^{t-1}(i)$ provides a good approximation of the optimal time-average cost. The error bound of the value iteration algorithm is established in~\cite{vi-average-MDP}.

\begin{algorithm}[t]
\caption{Value Iteration algorithm for the MTD game}
Input: $S, \underline{\tau}, \bar{\tau}, M, \alpha, \gamma, \epsilon, \delta.$ \\
Output: $\tau^{*} ,P^{*}.$
\begin{algorithmic}[1] 
\STATE $t = 0, V^{0}(i)=0,$ $\forall i\in S;$ 
\REPEAT
\STATE  $t=t+1$;
\FOR{$i\in S$}
\STATE $v=\infty;$
\FOR{$\tau=\underline{\tau};\tau\leq\bar{\tau};\tau=\tau+\delta$}
\STATE $\mathbf{p}' = \arg\min_{\mathbf{p}} V^{t}(i,\mathbf{p}, \tau)$;
\STATE $v'=V^{t}(i,\mathbf{p}', \tau)$
\IF{$v' < v$}
\STATE $\mathbf{p}^*_{i}=\mathbf{p}', \tau^{*}_i=\tau, v = v';$
\ENDIF
\ENDFOR
\STATE  $V^{t}(i)=v$;
\ENDFOR
\STATE $\overline{V}=\max_{i\in S}{|V^{t}(i)-V^{t-1}(i)|},$
\STATE $\underline{V}=\min_{i\in S}{|V^{t}(i)-V^{t-1}(i)|};$
\UNTIL {$\overline{V}-\underline{V} < \omega\underline{V}$}
\end{algorithmic}
\end{algorithm}

\subsection{Solving the Min-Max Problem}
A major obstacle in implementing Algorithm 1 is to find an efficient solution to the Min-Max problem~\eqref{min-max} for a fixed $\tau$ (line 7 in Algorithm 1). To this end, we first show that the optimal $\mathbf{p}$ to this problem has a simple structure, which significantly simplifies the problem. 

In the following discussion, we consider the Min-Max problem for configuration $i$ in iteration $t$ and for a fixed $\tau$. We drop the indices $i$ and $t$ to ease the notation. Let $m_j = m_{ij}, p_j = p_{ij}, w_j = w_{ij}$, and $V(j) = V^{t-1}(j)$. Let $w_{\min} = \min_{i \in S} w_{i}$, $w_{\max} = \max_{i \in S} w_{i}$, and $\rho = \frac{w_{\max}}{w_{\min}}$. Since the denominator in the first term and the second term in~\eqref{min-max} are both are constants, it suffices to consider the following problem:
\begin{align}
\min_{\mathbf{p}}\Big[
\max_{j}{(w_{j}p_{j})} + \sum_{j}{p_{j}(m_{j}+\gamma V(j))}\Big] \nonumber\\
\hspace{2ex} s.t. \ \ \mathbf{p} \in [\alpha,1]^{n}, \sum_j p_{j} = 1. \label{mm}
\end{align}

Let $U(\mathbf{p})$ denote the value of the objective function in~\eqref{mm} for a given $\mathbf{p}$. Let $\theta_{j}=m_{j}+\gamma V(j)$ denote the coefficient of $p_{j}$ in the second term of~\eqref{mm}. For a given $\mathbf{p}$, let $k$ be any configuration with $w_{k}p_{k}=\max_{j \in S}(w_{j}p_{j})$. We partition $S \backslash \{k\}$ into two sets where $A=\{a \in S:\theta_{a}> w_{k}+\theta_{k}\}$ and $B = S \backslash (A \cup \{k\})$. Let $\{b_j\}_{1 \leq j \leq |B|}$ denote the sequence of elements in $B$ sorted in $\theta$ non-decreasingly. 


\begin{prop}
For any optimal $\mathbf{p}$ to~\eqref{mm}, $p_{a}=\alpha, \forall a \in A$.
\end{prop}

\begin{proof}
Assume $p_{a}=\alpha+\epsilon$ for some $a\in A$ and $\epsilon>0$. We construct a new solution $\mathbf{p}'$ with $p'_{a} = \alpha$, $p'_{k}=p_{k}+\epsilon$, and $p'_{j}=p_{j}$ for any other $j$. Observe that 
$\mathbf{p}'$ is a feasible solution and $k=\arg\max_{j}(w_{j}p'_{j})$. It follows that $U(\mathbf{p}')-U(\mathbf{p}) = w_{k}(p'_{k}-p_{k})+(p'_{a}-p_{a})\theta_{a}+(p'_{k}-p_{k})\theta_{k} = (w_{k}+\theta_{k})\epsilon-\theta_{a}\epsilon < 0$ since $\theta_{a}> w_{k}+\theta_{k}$ for any $a \in A$. This contradicts the fact that $\mathbf{p}$ is an optimal solution. 
\end{proof}





\begin{prop}
Assume $\alpha \leq \frac{1}{n\rho}$. There is an optimal $\mathbf{p}$ to~\eqref{mm} where we can find an index $q \in \{1,2,...,|B|\}$ such that $p_{b_j} = \frac{w_{k}}{w_{b_j}}p_{k}$ for $1 \leq j \leq q$ and $p_{b_j}=\alpha$ for $q < j \leq |B|$.


\end{prop}

\begin{proof}
We first make the following observation. Consider any optimal solution $\mathbf{p}$ to~\eqref{mm}. Assume that there are $j_1, j_2 \in \{1,...|B|\}$ such that $j_1<j_2$, $p_{b_{j_1}} < \frac{w_{k}}{w_{b_{j_1}}}p_{k}$, and $p_{b_{j_2}}>\alpha$. We claim that we can construct a new optimal solution $\mathbf{p}'$ such that either $p'_{b_{j_1}} = \frac{w_{k}}{w_{b_{j_1}}}p'_{k}$ or $p'_{b_{j_2}}=\alpha$ (or both) while keeping other probabilities unchanged. To see this, let $\epsilon_1 = \frac{w_{k}}{w_{b_{j_1}}}p_{k}-p_{b_{j_1}}$ and $\epsilon_2 = p_{b_{j_2}}-\alpha$. We distinguish two cases. 

{\bf Case 1: $\epsilon_1 \leq \epsilon_2$:} we define $p'_{b_{j_1}} = p_{b_{j_1}}+\epsilon_1$, $p'_{b_{j_2}}=p_{b_{j_2}}-\epsilon_1$, and $p'_{j}=p_{j}$ for any other $j$. Observe that 
$\mathbf{p}'$ is a feasible solution and $p'_{b_{j_1}} = \frac{w_{k}}{w_{b_{j_1}}}p'_{k}$. Further, we still have $w_{k}p'_{k}=\max_{j}(w_{j}p'_{j})$. It follows that $U(\mathbf{p}')-U(\mathbf{p}) = (\theta_{b_{j_1}}-\theta_{b_{j_2}})\epsilon_1 \leq 0$ since $j_1<j_2$. Thus, $\mathbf{p}'$ is optimal. 

{\bf Case 2: $\epsilon_1>\epsilon_2$:} we define $p'_{b_{j_1}} = p_{b_{j_1}}+\epsilon_2$, $p'_{b_{j_2}}=p_{b_{j_2}}-\epsilon_2$, and $p'_{j}=p_{j}$ for any other $j$. 
$\mathbf{p}'$ is again feasible and $p'_{b_{j_2}} = \alpha$, and we still have $w_{k}p'_{k}=\max_{j}(w_{j}p'_{j})$. It follows that $U(\mathbf{p}')-U(\mathbf{p}) = (\theta_{b_{j_1}}-\theta_{b_{j_2}})\epsilon_2 \leq 0$ since $j_1<j_2$. Thus, $\mathbf{p}'$ is optimal. 

From the above observation, starting from any optimal solution $\mathbf{p}$, we can construct a new optimal solution $\mathbf{p}'$ in which there is an index $q \in \{1,2,...,|B|\}$ such that $p'_{b_j} = \frac{w_{k}}{w_{b_j}}p'_{k}$ for $1 \leq j < q$, $p_{b_j}=\alpha$ for $q < j \leq |B|$, and $p'_{b_q} \in [\alpha, \frac{w_{k}}{w_{b_q}}p'_{k}]$. We then show that, starting from such a $\mathbf{p}'$, we can construct a new optimal solution $\mathbf{p}''$ that satisfies the statement in the theorem. The main idea is to move a small amount of value from $\{p'_{b_j}, j < q\} \cup \{p'_{k}\}$ to $p'_{b_q}$ or the other way around depending on which direction is more beneficial. To this end, we again distinguish two cases:

{\bf Case 1:} $(1+\sum_{j< q}\frac{w_{k}}{w_{b_j}})\theta_{b_q}>\sum_{j<q}\frac{w_{k}}{w_{b_j}}\theta_{b_j}+(w_{k}+\theta_{k})$: Let $\epsilon>0$ be a small value to be determined. We construct a new solution $\mathbf{p}''$ where $p''_{k} =  p'_{k}+\epsilon$, $p''_{b_j} =  p'_{b_j}+\frac{w_{k}}{w_{b_j}}\epsilon$ for all $j<q$, $p''_{b_q} = p'_{b_q}- \epsilon-\sum_{j<q} \frac{w_{k}}{w_{b_j}}\epsilon$, and $p''_{b_j} = p'_{b_j}$ for other $j$. Note that $p''_{b_j} = \frac{w_{k}}{w_{b_j}}p''_{k}$ is maintained for $j<q$. Further, we can choose $\epsilon$ so that $p''_{b_q} = \alpha$. It is easy to see that $\mathbf{p}''$ is a feasible solution. Further, $U(\mathbf{p}'')-U(\mathbf{p}') = \Big(-(1+\sum_{j< q}\frac{w_{k}}{w_{b_j}})\theta_{b_q}+\sum_{j<q}\frac{w_{k}}{w_{b_j}}\theta_{b_j}+(w_{k}+\theta_{k})\Big)\epsilon \leq 0$. Hence, $\mathbf{p}''$ is also optimal.

{\bf Case 2:} $(1+\sum_{j< q}\frac{w_{k}}{w_{b_j}})\theta_{b_q}\leq\sum_{j<q}\frac{w_{k}}{w_{b_j}}\theta_{b_j}+(w_{k}+\theta_{k})$: We construct a new solution $\mathbf{p}''$ where $p''_{k} =  p'_{k}-\epsilon$, $p''_{b_j} =  p_{b_j}-\frac{w_{k}}{w_{b_j}}\epsilon$ for all $j<q$, $p''_{b_q} = p'_{b_q}+ \epsilon+\sum_{j<q} \frac{w_{k}}{w_{b_j}}\epsilon$, and $p''_{b_j} = p'_{b_j}$ for other $j$. We again have $p''_{b_j} = \frac{w_{k}}{w_{b_j}}p''_{k}$ for $j<q$, and we can choose $\epsilon$ so that $p''_{b_q} = \frac{w_{k}}{w_{b_q}}p''_{k}$. We claim that when $\alpha \leq \frac{1}{n\rho}$, we further have (i) $p''_{b_j} \geq \alpha$ for  $j \leq q$ and (ii) $p''_{j}w_{j}\leq w_{k} p''_{k}$ for $j \in A \cup \{b_{q+1},b_{q+2}, ..., b_{|B|}\}$. From these properties, we can conclude that $\mathbf{p}''$ is a feasible solution, and $U(\mathbf{p}'')-U(\mathbf{p}') = \Big((1+\sum_{j< q}\frac{w_{k}}{w_{b_j}})\theta_{b_q}-\sum_{j<q}\frac{w_{k}}{w_{b_j}}\theta_{b_j}-(w_{k}+\theta_{k})\Big)\epsilon \leq 0$. Thus, $\mathbf{p}''$ is also optimal. 

To prove claim (i), we have  
\begin{align*}
    p''_{b_j}-\alpha = & \frac{w_{k}}{w_{b_j}} \cdot \frac{1-(|A|+|B|-q)\alpha}{1+\sum_{j\leq q}\frac{w_{k}}{w_{b_j}}} - \alpha\\
    \geq & \frac{1}{w_{\max}} \cdot \frac{1-(|A|+|B|-q)\alpha}{\frac{1}{w_{k}}+\sum_{q \leq j}\frac{1}{w_{b_j}}} - \alpha\\
    = & \frac{1-(|A|+|B|-q)\alpha}{(q+1)\rho} - \alpha\\
    = & \frac{1-(n-1-q)\alpha}{(q+1)\rho} - \alpha \\
    = & \frac{1-[n-1-q+(q+1)\rho]\alpha}{(q+1)\rho} \\
    = & \frac{1-[n-1+\rho+q(\rho-1)]\alpha}{(q+1)\rho} \\
    \geq & \frac{1-[n-1+\rho+(n-1)(\rho-1)]\alpha}{n\rho} \\
    = & \frac{1-n\rho\alpha}{n\rho} \\
    \geq & 0
\end{align*}  
where the last inequality follows from $\alpha \leq \frac{1}{n\rho}$.

To prove claim (ii), it suffices to show that $p''_{k}w_{k} \geq \alpha w_{\max}$. We have
\begin{align*}
    p''_{k}w_{k}-\alpha w_{\max} &= \frac{1-(|A|+|B|-q)\alpha}{1+\sum_{j\leq q}\frac{w_{k}}{w_{b_j}}}w_{k} - \alpha w_{\max} \\ 
    &\geq \frac{1-(n-1-q)\alpha}{\frac{1}{w_{k}}+\sum_{j\leq q}\frac{1}{w_{m_j}}} - \alpha w_{\max}\\
    &\geq \frac{1-(n-1-q)\alpha}{1+q}w_{\min}-\alpha w_{\max} \\
    &= \frac{[1-(n-1-q)\alpha]w_{\min}-(1+q)\alpha w_{\max}}{1+q}\\
    &= \frac{w_{\min}[1-(n+(\rho-1)(1+q))\alpha]}{1+q}\\
    &\geq \frac{w_{\min}[1-(n+(\rho-1)n)\alpha]}{n}\\
    &= \frac{w_{\min}[1-\rho n \alpha]}{n}\\
    &\geq 0
\end{align*}
where the last inequality follows from $\alpha \leq \frac{1}{n\rho}$.

\end{proof}

\begin{algorithm}[t]
\caption{Solving the Min-Max problem~\eqref{mm}}
Input: $S, \{w_i\}, \{m_i\}, \gamma, V, \alpha.$ \\
Output: $\mathbf{p}^*.$
\begin{algorithmic}[1] 
\STATE $\theta_{j}= m_{j}+\gamma V(j)$, $\forall j\in S$;
\STATE $u=\infty$
\FOR{$k\in S$}
\STATE $p_{j}=\alpha$, for all $j$ such that $\theta_{j}> w_{k}+\theta_{k}$; 
\STATE $B=\{b:\theta_{b}\leq w_{k}+\theta_{k}\}$;
\STATE $\{b_j\} = $ the sequence of items in $B$ sorted in $\theta$ non-decreasingly; 
\FOR{$q=1; q\leq |B|; q=q+1$}
\STATE $p_{k}=\frac{1-|A|\alpha-(|B|-q+1)\alpha}{\sum_{j<q}{\frac{w_k}{w_j}}+1};$
\STATE $p_{b_j}=\frac{w_{k}}{w_{b_j}}p_{k}$, $\forall j \leq q$, $p_{b_j}=\alpha$, $\forall j> q$;
\IF{$w_{k}p_{k}+\sum_{j\in S} p_{j}\theta_{j} < u$}
\STATE $u = w_{k}p_{k}+\sum_{j\in S} p_{j}\theta_{j}$;
\STATE $\mathbf{p}^*= \mathbf{p}$;
\ENDIF
\ENDFOR
\ENDFOR
\end{algorithmic}
\end{algorithm}

Based on the two propositions above, we then design an efficient solution to~\eqref{mm} (see Algorithm 2). The algorithm iterates over all $k \in S$. For a given $k$, two sets $A$ and $B$ are identified and $p_j = \alpha$ for $j \in A$ (line 4). We then search for a proper index $q \leq |B|$ and set the value of $p_{b_j}$ for $b_j \in B$ according to  Proposition 2 (lines 7-9).
The running time of Algorithm 2 is dominated by sorting all the configurations according to their $\theta$ values for each $k$. Thus, the complexity of the algorithm is $O(n^2\log n)$, which is much faster than searching the whole probability space. 

\section{Numerical Results}\label{sec:eval}
\begin{figure*}[t]
\centering
\subfigure[]{
\label{Fig.n}
\includegraphics[width=2in]{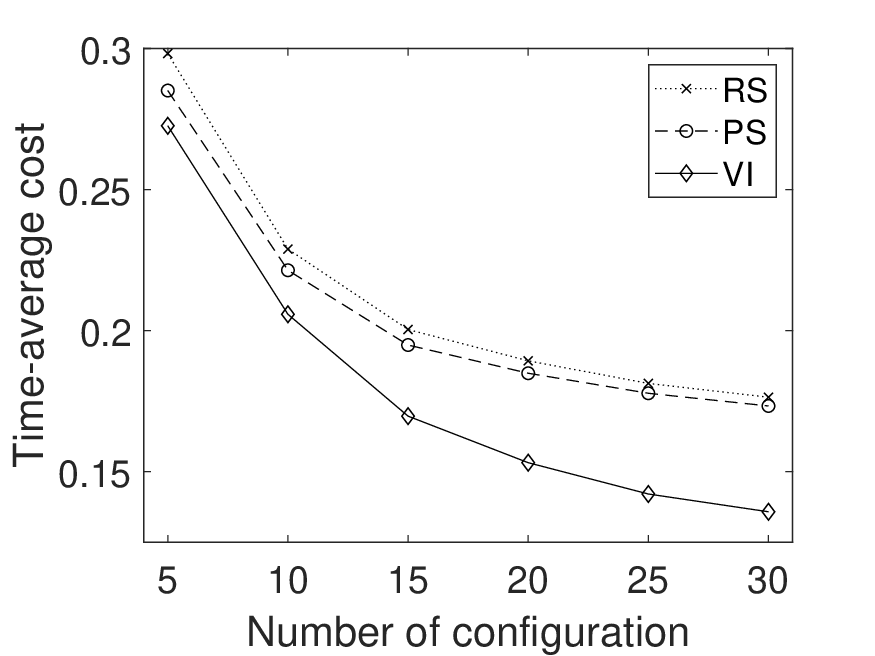}}
\subfigure[]{
\label{Fig.mu}
\includegraphics[width=2in]{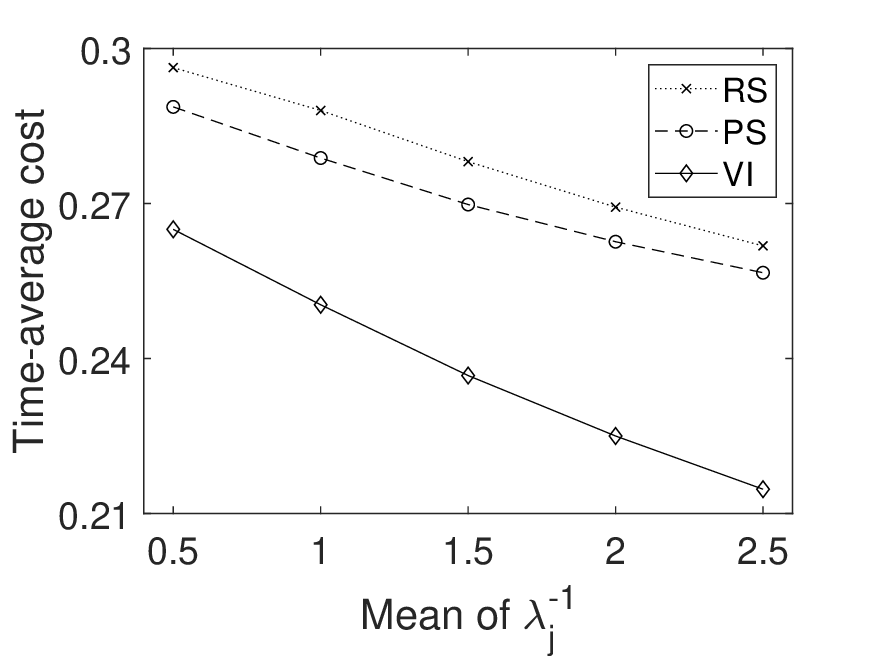}}
\subfigure[]{
\label{Fig.M}
\includegraphics[width=2in]{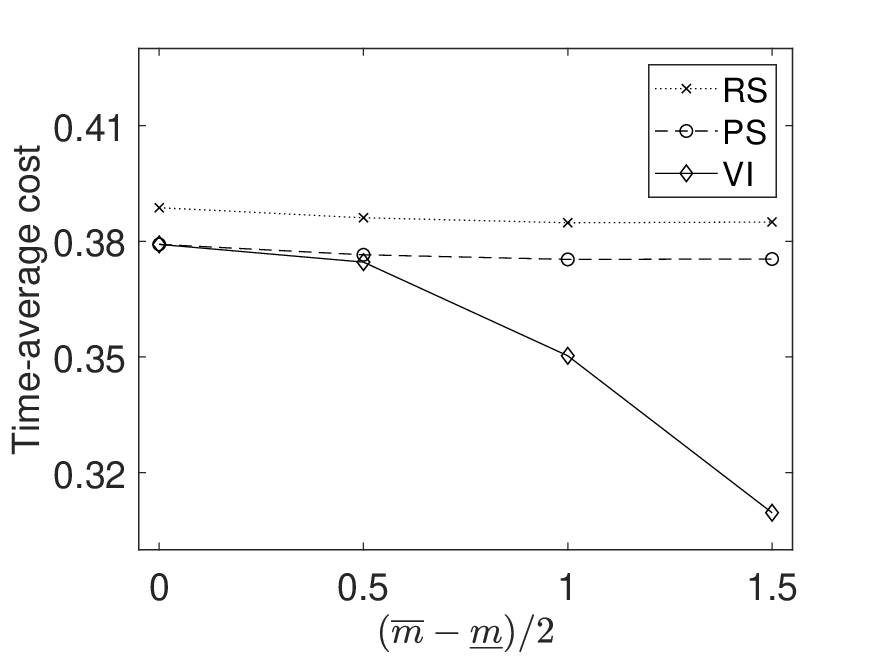}}
\caption{\small Simulation Results.}
\label{Fig.result}
\end{figure*}

In this section, we evaluate our MTD strategy through numerical studies under different system settings and demonstrate its advantage by comparing it with two heuristic strategies where a fixed defense period is used for all configurations: 



\begin{enumerate}
    \item Random sampling (RS): The defender stays in the current configuration for a fixed duration $\tau$ and then moves to a new configuration with probability $1/n$. The optimal $\tau$ is obtained by solving the following problem.
    \begin{equation*}
      \min_{\tau}\frac{\max_j\mathbb{E}(\max(\tau-a_j,0))+\frac{1}{n}\sum_{i,j}m_{ij}}{n\tau}
    \end{equation*}
    \item Proportional sampling (PS): The defender stays in the current configuration for a fixed duration $\tau$ and then moves to a new configuration $j$ with probability $p_j$ that is proportional to $w_j = \mathbb{E}[\max(\tau-a_j,0)]$. The defense strategy $(\tau, \{p_j\})$ is obtained by solving the following problem.
    \vspace{-1ex}
    \begin{alignat*}{2}
        \min_{\tau, \{p_j\}} && \frac{\max_{j}(w_{j}p_{j})+\sum_{i,j}p_{i}p_{j}m_{ij}}{\tau}\\
     s.t. && \hspace{2ex}
        w_j = \mathbb{E}(\max(\tau-a_j,0), \forall j\in S \\
        && w_{i}p_{i}=w_{j}p_{j}, \forall i,j\in S\\
        && \sum_{j\in S}p_{j} =1
    \end{alignat*}
\end{enumerate}

\vspace{1ex}
\noindent{\bf Simulation setup:} In the simulations, the number of configurations $n$ is chosen from $\{5, 10,...,30\}$. The set of migration costs are $i.i.d.$ samples from a uniform distribution. For each configuration $j$, its attack time $a_j$ follows an exponential distribution with parameter $\lambda_j$, where $\lambda_j$ is sampled from a uniform distribution and is $i.i.d.$ across $j$. In each simulation, we conduct 100 trials by taking 10 samples of the migration cost matrix $M$ and 10 samples of $\{\lambda_j\}$. We set $\alpha = 0.01$ and $\omega = 0.01$ in Algorithms 1 and 2. We set $\underline{\tau}=0.1, \bar{\tau} = 5$, and $\delta = 0.1$. For each $\tau_i \in \{0.1,0.2, ...,5\}$ and each $\lambda_j$, we estimate $w_{ij}$ by taking 500 samples of $a_j$, which are inputs to all the three policies.

\vspace{1ex}
\noindent{\bf Simulation results:} In Figure~\ref{Fig.result}(a), we evaluate the performance of the three policies by varying the number of configurations. The migration costs are sampled from $U(0,1.5)$ and $\lambda^{-1}_j$ are sampled from $U(1,2)$ for all $j$. 
We observe that for the all three strategies, the time-average cost decreases as the number of configurations increase. This is because the uncertainty to the attacker increases with $n$.  Moreover, 
the cost of VI decreases much faster than the two baselines, which indicates the weakness of the simple heuristics for large $n$. 

Figure~\ref{Fig.result}(b) compares the average costs of the three policies when the mean attack time $\lambda^{-1}_j$ is sampled from $U(\nu-0.5,\nu+0.5)$ for each $j$ where $\nu$ increases from 0.5 to 2.5. The number of configurations is fixed to 10 and the migration costs are sampled from $U(0.5,1)$. It is expected that the costs of all the strategies decrease as attack time increases. Again, our algorithm performs much better than the two baselines. Further, the gap increases for large $\nu$. This is because for large attack time, the migration cost becomes the dominant factor, which is not properly taken into account in the two baselines.

In Figure~\ref{Fig.result}(c), we compare the three policies under different variances of the migration cost distribution. In this case, $n=10$ and $\lambda^{-1}_j$ is sampled from $U(0.5,1.5)$ for each $j$. The migration costs are sampled from $U[\underline{m},\overline{m}]$ where the mean migration cost $(\underline{m}+\overline{m})/2$ is fixed to 1.5 and we vary $(\underline{m}-\overline{m})/2$. We observe that the performance of PS is close to VI for small variances while the gap becomes bigger for large variances. This is because when each node has a similar migration cost, the loss due to attacks becomes the dominant part in the total cost, which is considered in both PS and VI. On the other hand, when the variance becomes large, our algorithm is able to better handle the heterogeneity of configurations by jointly optimizing $P$ and $\tau$ and by considering a different $\tau_i$ for each $i$. 
\section{Conclusion}
\label{sec:conclusion}
In this paper, we propose a Stackelberg game model for moving target defense (MTD) that jointly considers the spatial and temporal decisions in MTD. In contrast to the $i.i.d.$ strategies considered in most previous works, our model considers the more general Markovian strategies and further incorporates state-dependent attack times. By formulating the defender's problem as a semi-Markovian decision process, we derive a nearly optimal defense strategy that can be efficiently implemented by utilizing the structure of the MTD game.    
\bibliographystyle{IEEEtran}
\bibliography{ref}
\end{document}